\documentclass[a4paper, 11pt]{article}
\usepackage[margin=1in]{geometry}
\usepackage{amsmath,amsthm,amssymb}
\usepackage{thm-restate}
\usepackage{enumitem}
\usepackage{hyperref}
\usepackage[dvipsnames]{xcolor}
\usepackage{tikz}
\usetikzlibrary{shapes}
\usepackage{algorithm}

\setlist{itemsep=0pt}
\hypersetup{
    colorlinks,
    linkcolor={RedViolet},
    citecolor={NavyBlue},
    urlcolor={NavyBlue}
}
\newcommand{\affiliation}[2]{%
    \footnote{#2}%
    \newcounter{#1}%
    \setcounter{#1}{\value{footnote}}
}
\newcommand{\reuse}[1]{%
    \footnotemark[\value{#1}]
}

\DeclareMathOperator{\Vol}{Volume}
\DeclareMathOperator{\Sample}{Sample}
\DeclareMathOperator{\Contains}{Contains}

\DeclareMathOperator{\Random}{Random}
\DeclareMathOperator{\Probe}{Probe}

\newcommand{\eps}{\varepsilon}
\newcommand{\Alg}{\mathcal{A}}
\newcommand{\cS}{\mathcal{S}}

\newcommand{\R}{\mathbb{R}}
\newcommand{\N}{\mathbb{N}}
\newcommand{\Z}{\mathbb{Z}}
\newcommand{\E}{\mathbb{E}}

\newtheorem{theorem}{Theorem}
\newtheorem{lemma}{Lemma}
\newtheorem{definition}{Definition}
\newtheorem{observation}{Observation}

\title{Approximating Klee's Measure Problem and\\a Lower Bound for Union Volume Estimation}
\date{}
\author{
Karl Bringmann\affiliation{uds}{Saarland University and Max-Planck-Institute for Informatics, Saarland Informatics Campus, Saarbr\"ucken, Germany. This work is part of the project TIPEA that has received funding from the European Research Council (ERC) under the European Unions Horizon 2020 research and innovation programme (grant agreement No. 850979).}\\\texttt{bringmann@cs.uni-saarland.de} \and
Kasper Green Larsen\affiliation{aarhus}{Aarhus University. Supported by a DFF Sapere Aude Research Leader Grant No. 9064-00068B.}\\\texttt{larsen@cs.au.dk} \and
Andr\'e Nusser\affiliation{cnrs}{Université Côte d'Azur, CNRS, Inria, France. This work was supported by the French government through the France 2030 investment plan managed by the National Research Agency (ANR), as part of the Initiative of Excellence of Université Côte d’Azur under reference number ANR-15-IDEX-01.
Part of this work was conducted while the author was at BARC, University of Copenhagen, supported by the VILLUM Foundation grant 16582.}\\\texttt{andre.nusser@cnrs.fr} \and
Eva Rotenberg\affiliation{dtu}{Technical University of Denmark. Supported by DFF Grant 2020-2023 (9131-00044B) ``Dynamic Network Analysis'', the VILLUM Foundation grant VIL37507 ``Efficient Recomputations for Changeful Problems'' and the Carlsberg Foundation Young Researcher Fellowship CF21-0302 ``Graph Algorithms with Geometric Applications''.}\\\texttt{erot@dtu.dk} \and
Yanheng Wang\reuse{uds}\\\texttt{yanhwang@cs.uni-saarland.de}
}

\begin{document}
\maketitle

\begin{abstract}
Union volume estimation is a classical algorithmic problem. Given a family of objects $O_1,\ldots,O_n \subset \mathbb{R}^d$, we want to approximate the volume of their union. In the special case where all objects are boxes (also called hyperrectangles) this is known as Klee's measure problem.
The state-of-the-art $(1+\varepsilon)$-approximation algorithm [Karp, Luby, Madras '89] for union volume estimation as well as Klee's measure problem in constant dimension $d$ uses a total of $O(n/\varepsilon^2)$ queries of three types: (i) determine the volume of $O_i$; (ii) sample a point uniformly at random from $O_i$; and (iii) ask whether a given point is contained in $O_i$.

First, we show that if an algorithm learns about the objects only through these types of queries, then $\Omega(n/\varepsilon^2)$ queries are necessary. In this sense, the complexity of [Karp, Luby, Madras '89] is optimal.
Our lower bound holds even if the objects are equiponderous axis-aligned polygons in $\mathbb{R}^2$, if the containment query allows arbitrary (not necessarily sampled) points, and if the algorithm can spend arbitrary time and space examining the query responses.

Second, we provide a more efficient approximation algorithm for Klee's measure problem, which improves the running time from $O(n/\varepsilon^2)$ to $O((n+\frac{1}{\varepsilon^2}) \cdot \log^{O(d)}(n))$. We circumvent our lower bound by exploiting the geometry of boxes in various ways:
(1) We sort the boxes into classes of similar shapes after inspecting their corner coordinates.
(2) With orthogonal range searching, we show how to sample points from the union of boxes in each class, and how to merge samples from different classes.
(3) We bound the amount of wasted work by arguing that most pairs of classes have a small intersection.
\end{abstract}

\newpage



\section{Introduction}
We revisit the classical problem of \emph{union volume estimation}:
given objects $O_1,\ldots,O_n \subset \mathbb{R}^d$, we want to estimate the volume of $O_1 \cup \ldots \cup O_n$.\footnote{Technically, the objects need to be measurable. In the most general form, $O_1,\ldots,O_n$ can be any measurable sets in a measure space, and we want to estimate the measure of their union. However, this work only deals with boxes in $\mathbb{R}^d$ (in our algorithm) and polygons in the plane (in our lower bound).}
This problem has several important applications such as DNF counting and network reliability; see the discussion in Section~\ref{sec:related work}.

The state-of-the-art solution~\cite{KarpLM89} works in a model where one has access to each input object $O_i$ by three types of queries: (i) determine the volume of the object, (ii) sample a point uniformly at random from the object, and (iii) ask whether a point is contained in the object.
Apart from these types of queries, the model allows arbitrary computations. The complexity of algorithms is thus measured by the number of queries.

After Karp and Luby~\cite{KarpL85} introduced this model, Karp, Luby and Madras~\cite{KarpLM89} gave an algorithm that $(1+\eps)$-approximates the volume of $n$ objects in this model with constant success probability.\footnote{The success probability can be boosted to $1-\delta$, adding a $\log(1/\delta)$ factor in time and query complexity.}
It uses $O(n/\eps^2)$ queries plus $O(n/\eps^2)$ additional time, which improves earlier algorithms \cite{KarpL85,LubyTechReport}, and only asks containment queries of previously sampled points.
The last 35 years have seen no algorithmic improvement. Is this classical upper bound best possible? We resolve the question in this work by providing a matching lower bound.

The union volume estimation problem was also studied recently in the streaming setting~\cite{MeelV021,Meel0V22}, where a stream of objects $O_1,\ldots,O_n \subset \Omega$ arrive in order. When we are at position $i$ in the stream, we can only query object $O_i$. This line of work yields algorithms with constant success probability that use $O(\textrm{polylog}(|\Omega|)/\eps^2)$ queries and additional time per object (and the same bound also holds for space complexity). Summed over all $n$ objects, the bounds match the classical algorithm up to the $\textup{polylog}(|\Omega|)$ factor. So, interestingly, even in the streaming setting the same upper bound can be achieved.\footnote{See also~\cite{TirthapuraW12} for earlier work studying Klee's measure problem in the streaming setting.}

The perhaps most famous application of the algorithm by Karp, Luby, and Madras~\cite{KarpLM89} is \emph{Klee's measure problem}~\cite{Klee77}.
This is a fundamental problem in computational geometry in which we are given $n$ axis-aligned boxes in $\R^d$ and want to compute the volume of their union.
An axis-aligned box is a set in the form $[a_1,b_1] \times \ldots \times [a_d,b_d] \subset \R^d$, and the input consists of the corner coordinates $a_1,b_1,\ldots,a_d,b_d$ of each box.
A long line of research on this problem and various special cases (e.g., for fixed dimensions or for cubes)~\cite{LeeuwenW81,OvermarsY91,Chan03,AgarwalKS07,Agarwal10,Chan10,YildizS12,Bringmann12} lead to an exact algorithm running in time $O(n^{d/2} + n \log n)$ for constant $d$~\cite{Chan13}.
A conditional lower bound suggests that any faster algorithm would require fast matrix multiplication techniques~\cite{Chan10}, but it is unclear how to apply fast matrix multiplication to this problem. 
On the approximation side, note that the three queries can be implemented in time $O(d)$ for any $d$-dimensional axis-aligned box. Thus the union volume estimation algorithm can be applied, and it computes a $(1+\eps)$-approximation to Klee's measure problem in time $O(nd/\eps^2)$, as has been observed in~\cite{BringmannF10}.
This direct application of union volume estimation was the state of the art for approximate solutions for Klee's measure problem until our work.
See Section~\ref{sec:related work} for interesting applications of Klee's measure problem.


\subsection{Our contribution} \label{sec:our contribution}


\subsubsection*{Lower bound for union volume estimation}
Given the state of the art, a natural question is to ask \emph{whether the query complexity of the general union volume estimation algorithm of \cite{KarpLM89} can be further improved}.
Any such improvement would speed up several important applications, cf.~Section~\ref{sec:related work}.
On the other hand, any lower bound showing that the algorithm of \cite{KarpLM89} is optimal also implies tightness of the known streaming algorithms (up to logarithmic factors), as the streaming algorithms match the static running time bound.


%

Previously, a folklore lower bound of $\Omega(n + 1/\eps^2)$ was known. But between this and the upper bound $O(n/\eps^2)$ there is still a large gap; consider for example the regime $\eps = 1/\sqrt{n}$ where the bounds are $\Omega(n)$ and $O(n^2)$, respectively. We close this gap by strengthening the lower bound to $\Omega(n/\eps^2)$, thereby also showing optimality of \cite{KarpLM89}. Note that the lower bound is about query complexity in the model; it does not make any computational assumption. In particular, it holds even if the algorithm spends arbitrary time and space in computation.

\begin{theorem}
\label{thm:lowerbound}
    Any algorithm that computes a $(1 + \eps)$-approximation to the volume of the union of $n$ objects via volume, sampling and containment queries with success probability at least $2/3$ must make $\Omega(n/\eps^2)$ queries.
\end{theorem}
We highlight that our lower bound holds even for equiponderous, axis-aligned polygons in the plane.

\subsubsection*{Upper bound for Klee's measure problem}
Our lower bound for union volume estimation implies that any improved algorithm for Klee's measure problem must exploit the geometric structure of boxes. To this end, we exploit our knowledge of the corner coordinates and classify the boxes by shapes. In addition, we apply the geometric tool of orthogonal range searching.
They allow us to approximate Klee's measure problem faster than the previous $O(n/\eps^2)$ algorithm. Here and throughout, we assume that the dimension $d$ is a constant; in particular, we suppress factors of $2^{O(d)}$ in the time complexity.

\begin{restatable}{theorem}{thmkmp}
\label{thm:kmp}
There is an algorithm that runs in time $O \left((n + \frac{1}{\eps^2}) \cdot \log^{2d+1} (n)\right)$ and with probability at least $0.9$ computes a $(1+\eps)$-approximation for Klee's measure problem.
\end{restatable}

This is a strict improvement when $\eps \leq 1/\log^{d+1}(n)$ is small. Consider for example the regime $\eps = 1/\sqrt{n}$: Our algorithm runs in near-linear time, whereas the previous algorithm runs in quadratic time. We remark that the success probability can be boosted to any $1-\delta$ by standard techniques which incurs an extra $\log(1/\delta)$ factor in the running time. Finally, we highlight that the core of our algorithm is an efficient method to sample uniformly and independently with a given density from the union of the input boxes. This can be of independent interest outside the context of volume estimation.

\subsection{Related work} \label{sec:related work}

A major application of union volume estimation is \emph{DNF counting}, in which we are given a Boolean formula in disjunctive normal form and want to count the number of satisfying assignments.
Computing the exact number of satisfying assignments is \#P-complete, so it likely requires superpolynomial time.
Approximating the number of satisfying assignments can be achieved by a direct application of union volume estimation, as described in \cite{KarpLM89}.
Their algorithm remains the state of the art for approximate DNF counting, see e.g. \cite{MeelSV19}.
This has been extended to more general model counting~\cite{PavanVBM21,ChakrabortyMV13,MeelSV19}, probabilistic databases~\cite{KimelfeldKS08,DalviS07,ReDS07}, and probabilistic queries on databases~\cite{CarmeliZBCKS22}.

We also mention \emph{network reliability} as another application for union volume estimation, which was already discussed in \cite{KarpLM89}. Additionally, Karger's famous paper on the problem~\cite{Karger99} uses the algorithm of \cite{KarpLM89} as a subroutine. However, the current state-of-the-art algorithms no longer use union volume estimation as a tool~\cite{CenHLP24}.

Finally, we want to draw a connection to the following well-known query sampling bound. Canetti, Even, and Goldreich~\cite{CanettiEG95} showed that approximating the mean of a random variable whose codomain is the unit interval requires $\Omega(1/\eps^2)$ queries, thus obtaining tight bounds for the sampling complexity of the mean estimation problem. Their bound generalize to $\Omega(1/(\mu \eps^2))$ on the number of queries needed to estimate the mean $\mu$ of a random variable in general. 
Before our work it was thus natural to expect that the $1/\eps^2$ dependence in the number of queries for union volume estimation is optimal. However, whether the factor $n$ is necessary, or the number of queries could be improved to, say, $O(n + 1/\eps^2)$, was open to the best of our knowledge.

Klee's measure problem is an important problem in computational geometry. One reason for its significance is that techniques developed for it can often be adapted to solve various related problems, such as the depth problem (given a set of boxes, what is the largest number of boxes that can be stabbed by a single point?)~\cite{Chan13} or Hausdorff distance under translation in $L_\infty$~\cite{Chan23}. Moreover, various other problems can be reduced to Klee's measure problem or to its related problems. For example, deciding whether a set of boxes covers its boundary box can be reduced to Klee's measure problem~\cite{Chan13}. The continuous $k$-center problem on graphs (i.e., finding centers that can lie on the edges of a graph that cover the vertices of a graph) is reducible to Klee's measure problem as well~\cite{ShiB08}. Finding the smallest hypercube containing at least $k$ points among $n$ given points can be reduced to the depth problem~\cite{EppsteinE94, Chan99, Chan13}. In light of this, it would be interesting to see if our approximation techniques generalize to any of these related problems.

\section{Technical overview} \label{sec:tech_overview}

We now sketch our upper and lower bounds at an intuitive level. The formal arguments will be presented in Section~\ref{sec:upperbound} and Section~\ref{sec:lowerbound}, respectively. In this paper we denote $[n] := \{1,2\dots,n\}$.

\subsection{Upper bound for Klee's measure problem}
Due to our lower bound, we have to exploit the structure of the input boxes to obtain a running time of the form $O \left( (n + \frac{1}{\eps^2}) \cdot \mathrm{polylog}(n) \right)$. 
We follow the common algorithmic approach of sampling. Specifically, we aim to sample a set $S$ from the union of boxes such that each point is selected with probability density $p$. For appropriately defined $p$, the value $|S|/p$ is a good estimate of the volume of the union. In the overview we assume that a proper $p$ is given and focus on the main difficulty: creating a sample set $S$ for the given $p$.

We start by sorting the input boxes into \emph{classes} by shape. Two boxes are in the same class if, in each dimension $k \in [d]$, their side lengths are both in $[2^{L_k}, 2^{L_k+1})$ for some $L_k \in \mathbb{Z}$.
We call two classes \emph{similar} if the corresponding side lengths are polynomially related (e.g., within a factor of $n^4$) in each dimension. We make two crucial observations:
\begin{enumerate}
    \item Each class is similar to only polylogarithmically many classes (Observation~\ref{obs:similar}).
    \item Dissimilar classes have a small intersection compared to their union (Observation~\ref{obs:dissimilar}, Figure~\ref{fig:dissimilar}).
\end{enumerate}

\begin{figure}[hb]
    \centering
\begin{tikzpicture}[thin,>=stealth, xscale=1, yscale=1]

    \draw[fill=blue!20!white,draw=blue!40!white] 
    (1,1) -- (1,1.3) -- (5,1.3) -- (5,1) -- (1,1);
    \draw[fill=blue!20!white,draw=blue!40!white] 
    (2,0) -- (2.7,0) -- (2.7,3) -- (2,3) -- (2,0);
    \draw[fill=none,draw=blue!40!white] 
    (1,1) -- (1,1.3) -- (5,1.3) -- (5,1) -- (1,1);  

    \draw[fill=blue!15!white,draw=blue!40!white] 
    (7.9,.2) -- (11.8,.2) -- (11.8,3.2) -- (7.9,3.2) -- (7.9,.2);
    \draw[fill=blue!20!white,draw=blue!40!white] 
    (8.7,2) -- (8.7,2.2) -- (12,2.2) -- (12,2) -- (8.7,2);
    \draw[fill=none,draw=blue!40!white] 
    (8.7,2) -- (11.8,2) -- (11.8,2.2) -- (8.7,2.2) -- (8.7,2);
\end{tikzpicture}
\caption{When the side lengths of two boxes differ a lot in at least one of their dimensions (in our examples, the $y$-axis), their intersection is small compared to their union.}
\label{fig:dissimilar}
\end{figure}

Our algorithm continues as follows. We scan through the classes in an arbitrary order. For each class, we sample with density $p$ from the union of the boxes in this class, but we only keep a point if it is not contained in any class that comes later in the order. To efficiently test containment in a later class, we use an orthogonal range searching data structure (with an additional dimension for the index of the class). When the scan finishes, we obtain a desired sample set $S$.

Let us elaborate why the algorithm is efficient.
\begin{description}
    \item[Sampling from a single class.]
    Our approach is simple yet powerful: (i) grid the space into cells of side lengths comparable to the boxes in this class; (ii) sample points from the relevant cells uniformly at random; (iii) discard points outside the union by querying an orthogonal range searching data structure. See Figure~\ref{fig:grid}. All these steps exploit the geometry of boxes. Since the grid cells and boxes have roughly the same shape, a significant fraction of the points sampled in (ii) are contained in the union, i.e., not discarded in (iii). The orthogonal range searching data structure allows us to quickly decide which points to discard. These explain the efficiency of sampling from a class.

    \item[Bounding the amount of wasted work.]
    Sampled points that appear in later classes also need to be discarded, and this is a potential (and only) source of inefficiency. So we need to bound the number of them. Such a point shows up later either (i) in a similar class, or (ii) in a dissimilar class. Our first observation stated that a class is similar to at most polylogarithmically many classes. So (i) can happen up to polylogarithmically many times from the perspective of a fixed point in the union. On the other hand, our second observation stated that the intersection of dissimilar classes is small. So on average (ii) rarely happens and does not affect the expected running time significantly.
\end{description}

\begin{figure}[hb]
    \centering
\begin{tikzpicture}[thin,>=stealth,
	vertex/.style={minimum size=0cm,circle,inner sep=2pt,fill=green!80!black,draw=none},
    cross/.style={cross out, draw, red,
         minimum size=2*(#1-\pgflinewidth), 
         inner sep=0pt, outer sep=0pt},
	every node/.style={vertex},
	scale = 1,xscale=1.35
	]

    \draw[fill=black!10!white,draw=none] 
    (0,0) -- (0,3) -- (3,3) -- (3,0) -- (0,0);
    \draw[fill=black!10!white,draw=none] 
    (4,1) -- (7,1) -- (7,4) -- (4,4) -- (4,1);

    \draw[step=1cm,gray,very thin] (-0,-0) grid (7,4);
 
    \draw[fill=blue!20!white,draw=blue!40!white] 
    (0.9,0.8) -- (2.5,0.8) -- (2.5,2.1) -- (0.9,2.1) -- (0.9,0.8);
    \draw[fill=blue!20!white,draw=blue!40!white] 
    (1.1,.5) -- (2.3,.5) -- (2.3,2.4) -- (1.1,2.4) -- (1.1,.5);
    \draw[fill=none,draw=blue!40!white] 
    (0.9,0.8) -- (2.5,0.8) -- (2.5,2.1) -- (0.9,2.1) -- (0.9,0.8);
    \draw[fill=blue!20!white,draw=blue!40!white] 
    (4.6,1.9) -- (6.1,1.9) -- (6.1,3.2) -- (4.6,3.2) -- (4.6,1.9);
 
    \draw[step=1cm,gray,very thin] (-0,-0) grid (6,4);

    \node[] at (1.25,1.55) {};
    \node[] at (1.55,1.29) {};
    \node[] at (1.5,2.25) {};
    \node[] at (5.22,2.75) {};
    \draw (.5,1.15) node[cross=3pt] {};
    \draw (2.77,.68) node[cross=3pt] {};
    \draw (4.3,2.45) node[cross=3pt] {};
    \draw (5.7,3.35) node[cross=3pt] {};
    \draw (6.2,1.35) node[cross=3pt] {};
    \draw (6.55,1.65) node[cross=3pt] {};
\end{tikzpicture}
\caption{We sample points in the \colorbox{gray!10!white}{grid cells} $\mathcal{G}$ that are intersected by a \colorbox{blue!10!white}{box}  $\mathcal{O}_i$ from a fixed class. We then use orthogonal range searching to determine whether a sampled point is in a box from the class and should be kept (\textcolor{green}{$\bullet$}), or is not and should be discarded (\textcolor{red}{$\times$}).}
\label{fig:grid}
\end{figure}

\subsection{Lower bound for union volume estimation}
Our lower bound is shown by a reduction from the Query-Gap-Hamming problem:
Given two vectors $x,y \in \{-1,+1\}^\ell$, distinguish whether their inner product is greater than $\sqrt{\ell}$ or less than $-\sqrt{\ell}$.
It is known that any algorithm distinguishing these two cases  with success probability at least $2/3$ must access $\Omega(\ell)$ bits of $x$ and $y$.

We first sketch why $\Omega(1/\varepsilon^2)$ queries are necessary to $(1+\eps)$-approximate the volume of the union of two objects in the query model.
Given a Query-Gap-Hamming instance $x,y$, we construct two objects $X=\{(j,x_j) : j \in [\ell]\}$ and $Y=\{(j,y_j) : j \in [\ell]\}$. These are discrete point sets, but they can be inflated to polygons such that cardinality corresponds to volume. See Figure~\ref{fig:simplecase} for an example.
Note that for all $k \in \{0, \dots, \ell\}$, we have
\[ |X \cup Y| = \ell + k \iff \langle x,y \rangle = \ell - 2k .\]
If an algorithm $\Alg$ can $(1+\eps)$-approximate $|X \cup Y|$, then it can approximate $k$ within additive error $2 \eps \ell$, thus it can also approximate $\langle x,y \rangle$ within additive error $4\eps\ell$. Setting $\ell = 1/(16\eps^{2})$, the additive error is at most $4\eps\ell = 4 \cdot \frac{1}{4\sqrt{\ell}} \cdot \ell = \sqrt{\ell}$, which suffices to distinguish $\langle x,y \rangle = \sqrt{\ell}$ from $\langle x,y \rangle = -\sqrt{\ell}$ and thereby solves the Query-Gap-Hamming instance.

Note that $|X| = |Y| = \ell$, so a volume query does not disclose any information about $x$ and $y$. Each sample or containment query accesses at most one bit of $x$ or $y$. Therefore, algorithm $\Alg$ has to make $\Omega(\ell) = \Omega(1/\eps^2)$ queries to $X,Y$.

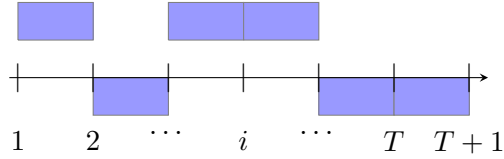
\begin{figure}[htbp]
    \centering
\begin{tikzpicture}[thin,>=stealth,
	vertex/.style={minimum size=0cm},
	every node/.style={vertex},
	scale = 0.5,xscale=.5
	]
	
    \draw[fill=blue!40!white,draw=gray] 
    (0,1) -- (4,1) -- (4,2) -- (0,2) -- (0,1);
    \draw[fill=blue!40!white,draw=gray] 
    (8,1) -- (12,1) -- (12,2) -- (8,2) -- (8,1);
    \draw[fill=blue!40!white,draw=gray] 
    (12,1) -- (16,1) -- (16,2) -- (12,2) -- (12,1);
    \draw[fill=blue!40!white,draw=gray] 
    (4,0) -- (8,0) -- (8,-1) -- (4,-1) -- (4,0);
    \draw[fill=blue!40!white,draw=gray] 
    (16,0) -- (20,0) -- (20,-1) -- (16,-1) -- (16,0);
    \draw[fill=blue!40!white,draw=gray] 
    (20,0) -- (24,0) -- (24,-1) -- (20,-1) -- (20,0);
 
	\node (neg) at (-1,0) {};
	\node[label={[yshift=-7]below:$1$}](0) at (0,0) {$\mid$};
	\node[label={[yshift=-7]below:$2$}](n) at (4,0) {$\mid$};
	\node[label={[yshift=-7]below:$\cdots$}](npn) at (8,0) {$\mid$};
	\node[label={[yshift=-7]below:$j$}](in) at (12,0) {$\mid$};
	\node[label={[yshift=-7]below:$\cdots$}](jn) at (16,0) {$\mid$};
	\node[label={[yshift=-7]below:$\ell$}](Tn) at (20,0) {$\mid$};
	\node[label={[yshift=-7]below:$\ell+1$}](Tp1n) at (24,0) {$\mid$};

	\draw[->] (neg) -> (25,0);
\end{tikzpicture}
    \caption{The vector $x=\left(+1,-1,+1,+1,-1,-1 \right)$ represented as the set $\{ (j,x_j) : j \in [6] \}$, where each point is drawn as a rectangle.\vspace{.5em}}
    \label{fig:simplecase}
\end{figure}

In order to generalize this lower bound for the union of two objects to an $\Omega(n/\varepsilon^2)$ lower bound for the union of $n$ objects, we need to ensure that each query gives away only $O(1/n)$ bits of information about $x$ and $y$.
We apply two obfuscations that jointly slow down the exposure of bits; see Figure~\ref{fig:polygons}.
Firstly, we introduce objects $X_1,\ldots ,X_n$ whose union is $X$ and objects $Y_1,\ldots ,Y_n$ whose union is $Y$.
Imagine cutting each $X$-induced rectangle in Figure~\ref{fig:simplecase} into $n$ side-by-side pieces and distributing them randomly among $X_1,\ldots , X_n$. Do the same for $Y$.
The idea is that one needs to make $\Omega(n)$ containment queries on the rectangular region in order to hit the correct piece; only then is the corresponding bit revealed.
Secondly, we introduce a large band shared by all $X_i$ and $Y_i$ for $i \in [n]$. In Figure~\ref{fig:polygons}, this is the long dark-blue rectangle that spans from left to right. Intuitively it enforces $\Omega(n)$ sample queries to obtain a single point that contains any information about $x$ and $y$.

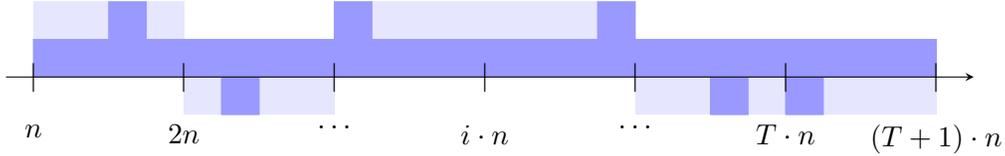
\begin{figure}[htbp]
    \centering
\begin{tikzpicture}[thin,>=stealth,
	vertex/.style={minimum size=0cm},
	every node/.style={vertex},
	scale = 0.5
	]
	
    \draw[fill=blue!10!white,draw=none] 
    (0,0) -- (4,0) -- (4,2) -- (0,2) -- (0,0);
    \draw[fill=blue!10!white,draw=none] 
    (8,0) -- (16,0) -- (16,2) -- (8,2) -- (8,0);
    \draw[fill=blue!10!white,draw=none] 
    (4,0) -- (8,0) -- (8,-1) -- (4,-1) -- (4,0);
    \draw[fill=blue!10!white,draw=none] 
    (16,0) -- (24,0) -- (24,-1) -- (16,-1) -- (16,0);

    \draw[fill=blue!40!white,draw=blue!40!white] (2,1) -- (3,1) -- (3,2) -- (2,2) -- (2,1);
    \draw[fill=blue!40!white,draw=blue!40!white] (8,1) -- (9,1) -- (9,2) -- (8,2) -- (8,1);
    \draw[fill=blue!40!white,draw=blue!40!white] (0,0) -- (24,0) -- (24,1) -- (0,1) -- (0,0);
    \draw[fill=blue!40!white,draw=blue!40!white] (5,-1) -- (6,-1) -- (6,0) -- (5,0) -- (5,-1);
	\draw[fill=blue!40!white,draw=blue!40!white] (15,1) -- (16,1) -- (16,2) -- (15,2) -- (15,1);
	\draw[fill=blue!40!white,draw=blue!40!white] (18,-1) -- (19,-1) -- (19,0) -- (18,0) -- (18,-1);
	\draw[fill=blue!40!white,draw=blue!40!white] (20,-1) -- (21,-1) -- (21,0) -- (20,0) -- (20,-1);

	\node (neg) at (-1,0) {};
	\node[label={[yshift=-5]below:$n$}](0) at (0,0) {$\mid$};
	\node[label={[yshift=-5]below:$2n$}](n) at (4,0) {$\mid$};
	\node[label={[yshift=-5]below:$\cdots$}] (npn) at (8,0) {$\mid$};
	\node[label={[yshift=-5]below:$jn$}] (in) at (12,0) {$\mid$};
	\node[label={[yshift=-5]below:$\cdots$}] (jn) at (16,0) {$\mid$};
	\node[label={[yshift=-5]below:$\ell n$}](Tn) at (20,0) {$\mid$};
	\node[label={[yshift=-5]below:$(\ell+1) n$}](Tp1n) at (24,0) {$\mid$};

	\draw[->] (neg) -> (25,0);

\end{tikzpicture}
    \caption{The vector $y$ or $x=\left(+1,-1,+1,+1,-1,-1 \right)$ gives rise to $n$ polygons; one of these polygons is illustrated in dark blue. The light blue area indicates the union of all these $n$ polygons.}
    \label{fig:polygons}
\end{figure}

\section{Approximation algorithm for Klee's measure problem} \label{sec:upperbound}

In this section, we present our approximation algorithm for Klee's measure problem in constant dimension $d$:
\thmkmp*

\subsection{Preliminaries}\label{sec:prelimsklee}

In Klee's measure problem we are given \emph{boxes} $O_1, \dots, O_n \subset \R^d$. Here, a box is an object of the form $O_i = [a_1,b_1] \times \ldots \times [a_d,b_d]$, and as input we are given the coordinates $a_1,b_1,\ldots,a_d,b_d$ of each box. Based on these, it is easy to compute the side lengths and volume of each box.

Throughout we write $V := \Vol(\bigcup_{i=1}^n O_i)$ for the volume of the union of boxes. The goal is to approximate $V$ up to a factor of $1+\eps$.
Our approach is based on sampling, so now let us introduce the relevant notions.

Recall the Poisson distribution $\mathrm{Pois}(\lambda)$ with mean and variance $\lambda$: It captures the number of active points in a universe, under the assumption that active points occur uniformly and independently at random across the universe, and that $\lambda$ points are active on average. The following definition is usually referred to as a homogeneous \emph{Poisson point process} at rate~$p$. Intuitively, we activate each point in some universe $U \subset \R^d$ independently with ``probability density'' $p$, thus the number of activated points follows the Poisson distribution with mean $p \cdot \Vol(U)$.

\begin{definition}[$p$-sample]
    Let $U \subset \R^d$ be a measurable set, and let $p \in [0,1]$. We say that a random subset $S \subseteq U$ is a \emph{$p$-sample of $U$} if for any measurable $U' \subseteq U$ we have that
    $|S \cap U'| \sim \mathrm{Pois}(p \cdot \Vol(U'))$.
\end{definition}
\noindent
In particular, if $S$ is a $p$-sample of $U$, then $|S| \sim \mathrm{Pois}(p \cdot \Vol(U))$. Two more useful properties follow from the definition:
\begin{enumerate}[label=(\roman*)]
    \item For any measurable subset $U' \subseteq U$, the restriction $S \cap U'$ is a $p$-sample of $U'$.
    \item The union of $p$-samples of two disjoint sets $U,U'$ is a $p$-sample of $U \cup U'$. 
\end{enumerate}

Besides sampling, we will apply orthogonal range searching to handle the so-called $\textsc{appears}(x,i)$ query: Given $x \in \R^d$ and $i \in \N$, is $x \in O_i \cup \cdots \cup O_n$?

\begin{lemma}
    \label{lem:appears}
    We can build a data structure in $O(n \log^{d+1}(n))$ time that answers $\textsc{appears}(x,i)$ queries in $O(\log^{d+1}(n))$ time.
\end{lemma}

\begin{proof}
    For each $j \in [n]$, map the box $O_j \subset \R^d$ to a higher-dimensional box
    \[ O_j^+ := O_j \times (-\infty,j] \subset \R^{d+1}. \]
    We then apply orthogonal range searching. Specifically, we build a multi-level segment tree over $\{O_1^+, \dots, O_n^+ \}$, which takes $O(n \log^{d+1}(n))$ time; see \cite[Section 10.4]{DBLP:books/lib/BergCKO08}. To answer the request $\textsc{appears}(x,i)$ where $x \in \R^d$ and $i \in \N$, we query the segment tree whether there exists a box $O_j^+$ that contains the point $(x,i)$; or phrased differently, whether $x \in O_j$ for some $j \geq i$. The query takes only $O(\log^{d+1}(n))$ time.
\end{proof}

For our algorithm to work, we also need a constant-factor approximation of the volume $V$. 
It is known that this can be computed in $O(n)$ time~\cite{KarpLM89}. In order to stay simple and self-contained, we state a weaker result by implementing the Karp-Luby algorithm \cite{KarpL85} with the help of \textsc{appears} queries.

\begin{lemma}[adapted from \cite{KarpL85}]
	\label{lem:crude}
    Given the data structure from Lemma~\ref{lem:appears}, there exists an algorithm that computes in $O(n \log^{d+1}(n))$ time a 2-approximation to $V$ with probability at least $0.9$.
\end{lemma}

\begin{algorithm}[h]
\caption{Crude approximation of Klee's measure}
\label{alg:crude}
\begin{enumerate}
	\item Compute prefix sums $S_j := \sum_{i=1}^j \Vol(O_i)$ for all $j \in \{0, \dots, n\}$.
	\item Initialize a counter $N := 0$.
	\item Repeat $40n$ times:
	\begin{itemize}
		\item Sample $u \in [0,1]$ uniformly at random. Binary search for the smallest $i$ such that $u \leq \frac{S_i}{S_n}$.
		\item Sample $x \in O_i$ uniformly at random.
		\item Increment $N$ if not $\textsc{appears}(x, i+1)$.
	\end{itemize}
	\item Output $\tilde{V} := \frac{N}{40n} \cdot S_n$.
\end{enumerate}
\end{algorithm}

\begin{proof}
    We claim that Algorithm \ref{alg:crude} has the desired properties.
    The time bound is easy to see: The computation of the prefix sums takes $O(n)$ time. In each iteration, the binary search costs $O(\log n)$ time, sampling of $x$ costs $O(1)$ time, and querying \textsc{appears} takes $O(\log^{d+1}(n))$ time. So in total we spend $O(n \log^{d+1}(n))$ time.
    
	For the correctness argument, we define two sets
	\begin{align*}
		P &:= \{(i,x) \;:\; i \in [n],\, x \in O_i\}, \\
		Q &:= \{(i,x) \;:\; i \in [n],\, x \in O_i \setminus (O_{i+1} \cup \ldots \cup O_n)\}.
	\end{align*}
 
    Consider an iteration in step 3. For any fixed value $j \in [n]$, we have
    \[ \Pr(i = j) = \Pr \left( \frac{S_{j-1}}{S_n} < u \leq \frac{S_j}{S_n} \right)
    = \frac{S_j - S_{j-1}}{S_n} = \frac{\Vol(O_j)}{S_n}. \]
    With this we can calculate the probability that the counter $N$ increments in this iteration:
    \begin{align*}
        \Pr((i,x) \in Q)
        &= \sum_{j=1}^n \Pr((i,x) \in Q \mid i=j) \cdot \Pr(i=j) \\
        &= \sum_{j=1}^n \frac{\Vol(O_j \setminus (O_{j+1} \cup \ldots \cup O_n))}{\Vol(O_j)} \cdot \frac{\Vol(O_j)}{S_n}
        = \frac{V}{S_n}.
    \end{align*}
    Since all iterations are independent, at the end of the algorithm we have $N \sim \mathrm{Bin}(40n, V/S_n)$. Hence $\widetilde{V}$ is an unbiased estimator for $V$.

	To analyze deviation, we observe that $V \geq \max_{i=1}^n \Vol(O_i) \geq S_n/n$. Therefore, $\E[N] = 40nV/S_n \geq 40$. Since $\E[N] \geq \mathrm{Var}[N]$, Chebyshev's inequality implies that
	\[
		\Pr \left( |N-\E[N]| \geq \frac{\E[N]}{2} \right)
		\leq \frac{4 \mathrm{Var}[N]}{(\E[N])^2}
		\leq \frac{4}{\E[N]}
		\leq 0.1.
	\]
	That is, with probability at least $0.9$ the output $\widetilde{V}$ is a 2-approximation to $V$.
\end{proof}

\subsection{Classifying boxes by shapes}

As our first step in the algorithm, we classify boxes by their shapes.
\begin{definition}
    Let $L_1,\ldots,L_d \in \Z$. 
    We say that a box $O \subset \R^d$ is of type $(L_1, \dots, L_d)$ if its side length in dimension $k$ is contained in $[2^{L_k}, 2^{L_k+1})$, for all $k \in [d]$.
\end{definition}

Using this definition, we partition the input boxes $O_1, \dots, O_n$ into \emph{classes} $C_1, \dots, C_m$ such that each class corresponds to one type of boxes. The notation is fixed from now on. For each $t \in [m]$, we denote $U_t := \bigcup_{O \in C_t} O \subset \R^d$, namely the union of boxes in class $C_t$.

Similar to $\textsc{appears}$, we may use orthogonal range searching to handle the so-called $\textsc{inClass}(x,t)$ query: Is a given point $x\in \mathbb{R}^d$ contained in~$U_t$?
\begin{lemma}
    \label{lem:in_class}
    We can build a data structure in $O(n\log ^{d+1} (n))$ time that answers $\textsc{inClass}(x,t)$ queries in $O(\log^{d+1}(n))$ time.
\end{lemma}

\begin{proof}
    Similar to the proof of Lemma \ref{lem:appears}, we transform each $O_i \in C_t$ to a higher-dimensional box
    \[ O_i \times \{t\} \subset \mathbb{R}^{d+1} \]
    and build a multi-level segment tree on top. The query $\textsc{inClass}(x,t)$ is thus implemented by querying the point $x\times \{t\} \in \mathbb{R}^{d+1}$ in the segment tree. 
\end{proof}

\subsubsection*{Sampling from a class}
The next lemma shows that we can $p$-sample from any $U_t$ efficiently by rejection sampling.
\begin{lemma}
    \label{lem:p-sample}
    Given $t \in [m]$, $p \in [0,1]$ and the data structure from Lemma~\ref{lem:in_class}, one can generate a $p$-sample of $U_t$ in expected time $O(|C_t| \log |C_t| + p \cdot \Vol(U_t) \cdot \log^{d+1}(n))$.
\end{lemma}
\begin{proof}
    Let $(L_1, \dots, L_d)$ be the type of class $C_t$.
    We subdivide $\R^d$ into the grid
    \[
        \mathcal{G}_{\infty} := \{ [i_1\,2^{L_1}, (i_1+1)\,2^{L_1}) \times \dots \times [i_d\,2^{L_d}, (i_d+1)\,2^{L_d}) \mid i_1, \dots, i_d \in \mathbb{Z} \}.
    \]
    We call each element of $\mathcal{G}_\infty$ a \emph{cell}. Let $\mathcal{G} := \{ G \in \mathcal{G}_\infty \mid G \cap U_t \neq \emptyset \}$ be the set of cells that intersect with $U_t$. Let $U := \bigcup_{G \in \mathcal{G}} G$.
    
    First we create a $p$-sample $S$ of $U$ as follows.
    Generate $K \sim \mathrm{Pois}(p\cdot\Vol(U))$, the number of points we are going to include. Then sample $K$ points uniformly at random from $U$ by repeating the following step $K$ times:
    Select a cell $G \in \mathcal{G}$ uniformly at random and then sample a point from $G$ uniformly at random.
    The sampled points constitute our set $S$.
    
    Next we compute $S \cap U_t$:
    For each $x \in S$, we query $\textsc{inClass}(x,t)$; if the answer is true then we keep $x$, otherwise we discard it.
    The resulting set $S \cap U_t$ is a $p$-sample of $U_t$, since restricting to a fixed subset preserves the $p$-sample property.

    Before we analyze the running time, we show that $U_t$ makes up a decent proportion of~$U$.
    Recall that every box in class $C_t$ is of type $(L_1, \dots ,L_d)$.
    Consider the projection to any dimension $k \in [d]$. Each projected box from $C_t$ can intersect at most three projected cells from $\mathcal{G}$. So each box from $C_t$ intersects at most $3^d$ cells from $\mathcal{G}$, implying that $|\mathcal{G}| \leq 3^d \, |C_t|$. Moreover, since the volume of any cell is at most the volume of a box in $C_t$, we have $\Vol(U) \leq 3^d \, \Vol(U_t)$.

    Recall that we assume $d$ to be constant and hence $3^d = O(1)$.
    The computation of~$\mathcal{G}$ takes $O(|\mathcal{G}| \log |\mathcal{G}|) \subseteq O(|C_t| \log |C_t|)$ time. The remaining time is dominated by the \textsc{inClass} queries. The expected size of $S$ is $p \cdot \Vol(U) \leq 3^d\,p\, \Vol(U_t)$. As we query \textsc{inClass} once for each point of $S$, the total expected time is $O(p \cdot \Vol(U_t) \cdot \log^{d+1}(n))$ by Lemma~\ref{lem:in_class}.
\end{proof}

\subsubsection*{Classes do not overlap much}
We show the following interesting property of classes, that the sum of their volumes is within a polylogarithmic factor of the total volume $V$.

\begin{lemma}
    \label{lem:quotient}
    We have $\sum_{t=1}^m \Vol(U_t) \leq 2^{3d+1} \log^d (n) \cdot V$.
\end{lemma}

We later exploit it to draw $p$-samples from $\bigcup_{i=1}^n O_i = \bigcup_{t=1}^m U_t$ efficiently.
Towards proving the lemma, let us collect some simple observations.

\begin{definition}
    We say that a class of type $(L_1,\dots,L_d)$ is \emph{similar} to a class of type $(L_1',\dots,L_d')$ if $2^{|L_k - L'_{k}|} < n^4$ for all $k \in [d]$. Otherwise they are said to be \emph{dissimilar}.
\end{definition}

\begin{observation}
    \label{obs:similar}
    Every class is similar to at most $8^d \log^d(n)$ classes.
\end{observation}

\begin{proof}
    Fix a type $(L_1,\dots,L_d)$. For each $k \in [d]$, there are at most $8 \log n$ many integers $L_k'$ such that $2^{|L_k - L_k'|} < n^4$.
\end{proof}

\begin{observation}
    \label{obs:dissimilar}
    If two boxes $O,O'$ are in dissimilar classes, then $\Vol(O \cap O') \leq 2V/n^4$.
\end{observation}

\begin{proof}
    Let $(L_1, \dots, L_d)$ be the type of $O$, and $(L'_1, \dots, L'_d)$ be the type of $O'$.
    Since the boxes belong to dissimilar classes, there is a dimension $k \in [d]$ such that $2^{|L_k - L_k'|} \geq n^4$.
    Without loss of generality, assume $2^{L_k - L_k'} \geq n^4$; the other case is symmetric.
    Let $[a_k,b_k]$ and $[a_k',b_k']$ be the projections of $O$ and $O'$ onto dimension $k$, respectively.
    Note that $b_k-a_k \in [2^{L_k},2^{L_k+1})$ and $b_k'-a_k' \in [2^{L_k'},2^{L_k'+1})$.
    So we have $\frac{b_k-a_k}{b'_k-a'_k} \geq 2^{L_k - (L_k' + 1)} \geq n^4/2$.
    In other words, at most a $2/n^4$ fraction of the interval $[a_k,b_k]$ intersects the interval $[a'_k,b'_k]$.
    Hence, $\Vol(O \cap O') \leq \Vol(O) \cdot 2/n^4 \leq 2V/n^4.$ \qedhere
\end{proof}

\begin{proof}[Proof of Lemma~\ref{lem:quotient}]
    Without loss of generality assume $\Vol(U_1) \geq \cdots \geq \Vol(U_m)$.
    We construct a set of indices $T \subseteq [m]$ by the following procedure:
    \begin{itemize}
        \item Initialize $T = \emptyset$.
        \item For $t = 1, \dots, m$, if $C_t$ and $C_s$ are dissimilar for all $s \in T$, then add $t$ to $T$.
    \end{itemize}
    
    For $t \in [m]$, we have $t \not\in T$ only if there exists an $s \in T$ such that $C_s,C_t$ are similar and $\Vol(U_s) \geq \Vol(U_t)$; we call $s$ a \emph{witness} of $t$. If multiple witnesses exist, then we pick an arbitrary one. Conversely, every $s \in T$ can be a witness at most $8^d \log^d(n)$ times by Observation~\ref{obs:similar}. Hence,
    \begin{equation} \label{eq:1}
    \sum_{t=1}^m \Vol(U_t) \le 8^d \log^d (n) \cdot \sum_{t \in T} \Vol(U_t).
    \end{equation}

    It remains to bound $\sum_{t \in T} \Vol(U_t)$. Consider any distinct $s,t \in T$. By construction, $C_s$ and $C_t$ are dissimilar; and each class contains at most $n$ boxes. So $\Vol(U_s \cap U_t) \leq n^2 \cdot (2V/n^4) = 2V/n^2$ by Observation \ref{obs:dissimilar}. Using this and inclusion-exclusion, we bound
    \begin{align*}
        \sum_{t \in T} \Vol(U_t)
        & ~\leq~ \Vol\left( \bigcup_{t \in T} U_t \right) + \sum_{\{s,t\} \subseteq T} \Vol(U_s \cap U_t) \\
        & ~\leq~ V + \binom{m}{2} \, \frac{2V}{n^2} ~\leq~ 2V.
    \end{align*}
    Plugging this into the right-hand side of Expression~(\ref{eq:1}), we obtain the lemma statement.
\end{proof}

\subsection{Joining the classes}
    Recall that $C_1, \dots, C_m$ are the classes of the input boxes and $U_1, \dots, U_m$ their respective unions. Assume without loss of generality that the boxes are ordered in accordance with the class ordering, that is, $C_1 = \{O_1, \cdots, O_{i_1}\}$ form the first class, $C_2 = \{O_{i_1+1}, \cdots, O_{i_2}\}$ form the second class, and so on. In general, we assume that $C_t = \{O_{i_{t-1}+1}, \ldots, O_{i_t}\}$ for all $t \in [m]$, where $0 = i_0 < i_1 < \ldots < i_m = n$.
    
    Let $D_t := U_t \setminus (\bigcup_{s=t+1}^m U_s)$ be the points in $U_t$ that are not contained in later classes. Note that $D_1, \dots, D_m$ is a partition of $\bigcup_{t=1}^m U_t = \bigcup_{i=1}^n O_i$. Hence, to generate a $p$-sample of $\bigcup_{i=1}^n O_i$, it suffices to draw $p$-samples from each $D_t$ and then take their union.\footnote{This idea has previously been used on objects, by considering the difference $D'_i := O_i \setminus (\bigcup_{j=i+1}^n O_j)$~\cite{KarpL85, MeelV021}, while we use this idea on classes.}
    To this end, we draw a $p$-sample $S_t$ from $U_t$ via Lemma \ref{lem:p-sample}. Then we remove all $x \in S_t$ for which $\textsc{appears}(x, i_t + 1)=\text{true}$; these are exactly the points that appear in a later class. What remains is a $p$-sample of $D_t$. 
    The union of these sets thus is a $p$-sample of $\bigcup_{i=1}^n O_i$, and we can use the size of this $p$-sample to estimate the volume $V$ of $\bigcup_{i=1}^n O_i$. The complete algorithm is summarized in Algorithm \ref{alg:kmp}.
    
    \begin{algorithm}[ht]
    \caption{Approximation of Klee's measure}
    \label{alg:kmp}
    \begin{enumerate}
        \item Partition the boxes into classes $C_1, \dots, C_m$. Relabel the boxes so that their indices are in accordance with the class ordering, i.e., $C_t = \{ O_{i_{t-1}+1}, \ldots, O_{i_t} \}$ for all $t \in [m]$.
        \item Build the data structures from Lemmas~\ref{lem:appears} and \ref{lem:in_class}.
        \item Obtain a crude estimate $\widetilde{V}$ by Lemma \ref{lem:crude}. Set $p := 8/(\eps^2 \widetilde{V})$.
        \item For $t = 1, \dots, m$ do:
        \begin{itemize}
            \item Draw a $p$-sample $S_t$ from the union $U_t = \bigcup_{O \in C_t} O$ via Lemma \ref{lem:p-sample}.
            \item Compute $|S_t'|$ where $S_t' := \{x \in S_t : \textsc{appears}(x,i_t+1) = \text{false} \}$.
        \end{itemize}
        \item Output $\sum_{t=1}^m |S_t'|/p$.
    \end{enumerate}
    \end{algorithm}
	
	\begin{lemma}
        \label{lem:kmp-quality}
		Conditioned on $\widetilde{V} \leq 2V$, Algorithm \ref{alg:kmp} outputs a $(1+\eps)$-approximation to $V$ with probability at least $3/4$.
	\end{lemma}

	\begin{proof}
        Note that for all $t \in [m]$, the set $S_t'$ is a $p$-sample of $D_t$. Since $D_1, \dots, D_m$ partition $\bigcup_{t=1}^m U_t = \bigcup_{i=1}^n O_i$, their union $\bigcup_{t=1}^m S_t'$ is a $p$-sample of $\bigcup_{i=1}^n O_i$. It follows that $N := \sum_{t=1}^m |S_t'| \sim \mathrm{Pois}(pV)$.
        The expectation and variance of $N$ are $pV = 8V/(\eps^2\widetilde{V}) \geq 4/\eps^2$. So by Chebyshev,
        \[
            \Pr(|N - pV| > \eps pV )
            \leq \frac{\mathrm{Var}[N]}{(\eps pV)^2}
            \leq \frac{1}{4}.
        \]
        Hence, with probability at least $3/4$, the output $N/p$ is a $(1+\eps)$-approximation to~$V$.
	\end{proof}

    \begin{lemma}
        \label{lem:kmp-time}
        Conditioned on $\widetilde{V} \geq \frac{V}{2}$, Algorithm \ref{alg:kmp} runs in time $O \left((n + \frac{1}{\eps^2}) \cdot \log^{2d+1} (n)\right)$ in expectation.
    \end{lemma}

    \begin{proof}
        Step 1 takes $O(n \log n)$ time: we can compute the side lengths of each box, determine its class, then sort the boxes by class index.
        Step 2 takes $O(n \log^{d+1}(n))$ time by Lemmas~\ref{lem:appears} and \ref{lem:in_class}.
        Step 3 takes $O(n \log^{d+1}(n))$ time by Lemma~\ref{lem:crude}.

        In step 4, iteration $t$, sampling $S_t$ costs expected time $O((i_t-i_{t-1}) \log (i_t-i_{t-1}) + p \Vol(U_t) \cdot \log^{d+1}(n))$ by Lemma \ref{lem:p-sample}, and computing $S_t'$ costs expected time $O((1 + p \Vol(U_t)) \cdot \log^{d+1}(n))$ by Lemma \ref{lem:appears}. Over all iterations, the expected running time is
        \[ O \left( \log^{d+1} (n) \cdot \left(n + p \, \sum_{t=1}^m \Vol(U_t) \right) \right). \]
        Substituting $p = 8 / (\eps^2 \tilde{V}) \leq 16 / (\eps^2 V)$ and applying Lemma~\ref{lem:quotient}, we can bound
        \[
            p \, \sum_{t=1}^m \Vol(U_t)
            \leq \frac{16}{\eps^2 V} \, \sum_{t=1}^m \Vol(U_t)
            \leq \frac{2^{3d+5} \log^d(n)}{\eps^2}.
        \]
        Hence the expected running time of step 4 is $O \left(\log^{2d+1} (n) \cdot (n + \frac{1}{\eps^2})\right)$.
        
        Finally, step 5 takes $O(n)$ time.
    \end{proof}
    
    \begin{proof}[Proof of Theorem \ref{thm:kmp}]
        Run Algorithm \ref{alg:kmp} with a time budget tenfold the bound in Lemma~\ref{lem:kmp-time}; abort the algorithm as soon as it spends more time than the budget allows. So the stated time bound is clearly satisfied.
        Now consider three bad events:
        \begin{itemize}
            \item $\widetilde{V} \not\in [\frac{V}{2}, 2V]$.
            \item $\widetilde{V} \in [\frac{V}{2}, 2V]$, but the algorithm is aborted.
            \item $\widetilde{V} \in [\frac{V}{2}, 2V]$ and the algorithm is not aborted, but it does not output a $(1+\eps)$-approximation to $V$.
        \end{itemize}
        By Lemma~\ref{lem:crude}, the first event happens with probability at most $0.1$.
        By Lemma~\ref{lem:kmp-time} and Markov's inequality, the second event happens with probability at most $0.1$. Lastly, by Lemma \ref{lem:kmp-quality}, the third event happens with probability at most $1/4$. So the total error probability is at most $0.1+0.1+\frac{1}{4} = \frac{9}{20}$. If none of the bad events happen, then the algorithm correctly outputs a $(1+\eps)$-approximation to $V$. The success probability of $\frac{11}{20}$ can be boosted to, say, $0.9$ by returning the median of a sufficiently large constant number of repetitions.  
    \end{proof}

\section{Lower bound for union volume estimation}\label{sec:lowerbound}
In this section, any randomized algorithm or protocol is assumed to have success probability at least $2/3$.
We consider the discrete version of union volume estimation in two dimensions, where each object $O_i$ is a finite subset of the integer lattice $\Z^2$. The goal is to $(1+\eps)$-approximate the cardinality $|O_1 \cup \ldots \cup O_n|$ of their union. The algorithm does not have direct access to the objects, but it can ask three forms of queries:
\begin{itemize}
    \item $\Vol(O_i)$: Return the cardinality $|O_i|$.
    \item $\Sample(O_i)$: Draw a uniform random point from $O_i$.
    \item $\Contains((a,b),O_i)$: Given a point $(a,b) \in \Z^2$, return whether $(a,b) \in O_i$ or not.
\end{itemize}
We aim to prove that $\Omega(n/\eps^2)$ queries are necessary. Later in Section \ref{sec:discrete-continuous} we translate this to a lower bound in the continuous space $\R^2$, thereby proving Theorem~\ref{thm:lowerbound}.

Our starting point is the Query-Gap-Hamming problem: The inputs are two (hidden) vectors $x,y \in \{-1,1\}^\ell$ and we can access one bit of $x$ or $y$ of our choice at a time. The goal is to distinguish the cases $\langle x, y \rangle \geq \sqrt{\ell}$ and $\langle x, y \rangle \leq -\sqrt{\ell}$ using as few accesses as possible. The following lemma is folklore.

\begin{lemma}
    \label{lem:gap-hamming}
    There exists a constant $\delta \in (0,1)$ with the following property. For every $\ell \in \N$, any randomized algorithm for Query-Gap-Hamming on vectors of length $\ell$ needs at least $\delta\ell$ accesses in expectation.
\end{lemma}

\begin{proof}
    We reduce from Gap-Hamming, a communication problem where Alice holds a vector $x \in \{-1,1\}^\ell$, Bob holds a vector $y \in \{-1,1\}^\ell$, and their goal is to distinguish $\langle x, y \rangle \geq \sqrt{\ell}$ from $\langle x, y \rangle \leq -\sqrt{\ell}$ while communicating as few bits as possible. It is known that the two-way, public-coin randomized communication complexity of Gap-Hamming is $\Omega(\ell)$~\cite{ChakrabartiR11}. In verbose, there exists a constant $\delta \in (0,1)$ such that, for every $\ell \in \N$, any public-coin randomized protocol for Gap-Hamming needs at least $\delta\ell$ bits of communication.

    Now consider an arbitrary randomized algorithm $\cS$ for Query-Gap-Hamming. We transform it to a protocol between Alice and Bob. The two parties simulate $\cS$ synchronously, sharing a random tape. When $\cS$ tries to access $x_j$, Alice sends the bit to Bob. When $\cS$ tries to access $y_j$, Bob sends the bit to Alice. In the end $\cS$ outputs a number, which Alice and Bob use as the output of the protocol. The number of communicated bits is equal to the number of accesses by $\cS$. Since the former is at least $\delta \ell$, the latter needs to be so as well.
\end{proof}

From now on, let us fix the constant $\delta \in (0,1)$ guaranteed by Lemma \ref{lem:gap-hamming}. Let $n \in \N$, $\eps \in (0,1)$ be arbitrary and define $\ell := \frac{1}{36\eps^2}$. We reduce Query-Gap-Hamming on vectors of length $\ell$ to estimating the cardinality of the union of $2n$ objects in $\Z^2$.

\subsubsection*{The reduction}
From the hidden vectors $x,y \in \{-1,1\}^\ell$ we construct $2n$ objects $X_1, \dots, X_n, Y_1, \dots, Y_n \subset \Z^2$.
Let $R := \{(n+1,0), \dots, (n\ell+n,0)\}$. Take independent random permutations $\pi_1,\dots,\pi_\ell$ of $[n]$ and define
\[ X_i := R \cup \{(jn + \pi_j(i), x_j) : j \in [\ell]\} \]
for $i \in [n]$. Then take another set of independent random permutations $\tau_1,\dots,\tau_\ell$ and define
\[ Y_i := R \cup \{(jn + \tau_j(i), y_j) : j \in [\ell]\} \]
for $i \in [n]$. Note that $R$ is a subset of all $X_i$ and $Y_i$.

How is $\langle x,y \rangle$ related to the cardinality of the union $|X_1 \cup \cdots \cup X_n \cup Y_1 \cup \cdots \cup Y_n|$? Consider an arbitrary index $j \in [\ell]$. If $x_j = y_j$ then the point sets $\{(jn + \pi_j(i), x_j) : i \in [n]\}$ and $\{(jn + \tau_j(i),y_j) : i \in [n]\}$ are equal (regardless of the concrete permutations), so they together contribute $n$ to the cardinality of the union. On the other hand, if $x_j \neq y_j$ then they are disjoint and thus contribute $2n$. Furthermore, the point set $R$ is contained in all objects and contributes $n\ell$. Hence, the cardinality of the union is equal to
\[
    n\ell + \sum_{j : x_j = y_j} n + \sum_{j : x_j \neq y_j} 2n = \frac 52 n\ell - \frac 1 2 n \cdot \left(\sum_{j : x_j = y_j} 1 + \sum_{j : x_j \neq y_j} (-1) \right) = \frac 52 n\ell - \frac 12 n\langle x, y \rangle.
\]
Suppose we get a $(1 + \eps)$-approximation $\rho$ to the cardinality of the union, then
\[
    \rho \in \left[(1 - \eps) (\tfrac{5}{2} n\ell - \tfrac{1}{2} n\langle x, y \rangle), \, (1+\eps)(\tfrac{5}{2} n\ell - \tfrac{1}{2} n\langle x, y \rangle) \right].
\]
Since $|\langle x,y \rangle|\leq \ell$, by computing $(\frac 52 n\ell - \rho)\cdot \frac 2n$ we obtain a value in $\langle x, y \rangle \pm 6\eps \ell$, that is, an additive $6 \eps \ell$ approximation to $\langle x, y \rangle$. Our choice of $\eps = \frac{1}{6 \sqrt{\ell}}$ allows to distinguish $\langle x, y \rangle \geq \sqrt{\ell}$ from $\langle x, y \rangle \leq -\sqrt{\ell}$.

The reduction is not yet complete. The catch is that an algorithm for union volume estimation can only learn about the objects via queries. It is left to \emph{us} to answer those queries. In Algorithm~\ref{alg:oracle} we describe how to answer queries on object $X_i$. Queries on object $Y_i$ can be handled similarly by replacing $\pi_j$ with $\tau_j$ and $x_j$ with $y_j$. The correctness is easy to verify, and this completes the reduction. So far, the argument works for arbitrary permutations; randomness becomes relevant when we bound the number of bit accesses.
\begin{algorithm}[ht]
\caption{Answering queries}
\label{alg:oracle}
\begin{itemize}
    \item For query $\Vol(X_i)$: Return $(n+1) \ell$.
    \item For query $\Sample(X_i)$: Draw a random block $j \in [\ell]$ and a random shift $s \in [n]$. With probability $1 - \frac{1}{n + 1}$ return $(jn + s, 0)$. Otherwise return $(jn + \pi_j (i), x_j)$.
    \item For query $\Contains((a, b), X_i)$: If $(a,b) \in R$ then return true. Else, compute $j := \lfloor a / n \rfloor$ and $s := a - jn$. If $j \not\in [\ell]$ or $\pi_j (i) \neq s$ then return false. Otherwise, if $x_i = b$ then return true; else return false.
\end{itemize}
\end{algorithm}

\subsubsection*{Bounding the number of accesses}
Now for the sake of contradiction, suppose there exists a randomized algorithm $\Alg$ that $(1 + \eps)$-approximates the cardinality of the union of any $2n$ objects in $\Z^2$, using less than $\frac{\delta}{360 \cdot 2^{9/\delta}} \cdot \frac{n}{\eps^2}$ queries. We run $\Alg$ on the objects defined earlier (using Algorithm~\ref{alg:oracle} to handle its queries). From the output, we can distinguish $\langle x,y \rangle \geq \sqrt{\ell}$ and $\langle x,y \rangle \leq -\sqrt{\ell}$.
We claim that less than $\delta\ell$ bits of $x,y$ are accessed from the side of Algorithm~\ref{alg:oracle}.

To this end, note that accesses to $x,y$ only happen in the ``otherwise'' clauses of $\Sample()$ and $\Contains()$ in Algorithm~\ref{alg:oracle}. For analysis we isolate the following mini game:

\begin{lemma}
    \label{lem:mini-game}
    Consider a game of $m := \lfloor 2^{-9/\delta} n \rfloor$ rounds. At the start, Merlin generates a secret random permutation $\pi$ of $[n]$. In each round, Arthur adaptively sends one of the following two requests to Merlin:
    \begin{itemize}
        \item $\Random()$: Answer ``yes'' with probability $\frac{1}{n+1}$ and ``no'' with the remaining probability.
        \item $\Probe(i,s)$: Answer ``yes'' if $\pi(i) = s$ and ``no'' otherwise.
    \end{itemize}
    
    No matter what strategy Arthur employs, the probability that Merlin ever answers ``yes'' in the game is at most $2\delta/5$.
\end{lemma}

\begin{proof}
    Let $E$ be the event that Merlin ever answers ``yes'' \emph{and} the first one was due to a $\Random()$ request. Let $F$ be the event that Merlin ever answers ``yes'' \emph{and} the first one was due to a $\Probe()$ request. We will bound $\Pr(E)$ and $\Pr(F)$ separately.

    To bound $\Pr(E)$, define events $E_1, \dots, E_m$ where $E_t$ indicates that Arthur's $t$-th request is $\Random()$ and the answer is ``yes''. Clearly $\Pr(E_t) \leq \frac{1}{n+1}$, thus $\Pr(E) \leq \sum_{t=1}^m \Pr(E_t) \leq \frac{m}{n+1} < 2^{-9/\delta} < \delta/15$.

    To bound $\Pr(F)$ we use an entropy argument. In particular, we will describe and analyze an encoding scheme for a random permutation $\pi$ of $[n]$.
    
    The encoder receives $\pi$ and acts as Merlin in the game, using $\pi$ as the secret permutation. It interacts with Arthur till the end of the game. If $F$ does not happen then the encoding is simply a bit \texttt{0} followed by $\pi$. On the other hand, if $F$ happens, then consider the first round $T \in [m]$ when ``yes'' appeared; denote that request by $\Probe(i,s)$. The encoding starts with a bit \texttt{1}, then the number $T$, and finishes with the ordering $\pi'$ of $[n] \setminus \{i\}$ induced by $\pi$.

    The decoder receives the encoding and tries to reconstruct $\pi$. First it reads the leading bit of the encoding. If the bit is \texttt{0} then $\pi$ immediately follows. If the bit is \texttt{1} then the decoder recovers the number $T$. It pretends to be Merlin and starts a game with Arthur, using the same random tape as the encoder but without knowledge of $\pi$. For the first $T-1$ requests it just answers ``no'' to Arthur. Now comes the $T$-th request from Arthur, which must be a $\Probe(i,s)$ request with answer ``yes''. This allows the decoder to deduce $\pi(i) = s$. The remaining entries in $\pi$ can be fully restored from $\pi'$.

    Writing $p := \Pr(F)$, the encoding length is at most
    \begin{align*}
        & (1 - p) \Big( 1 + \lceil \log(n!) \rceil \Big) + p\Big(1 + \lceil \log(m) \rceil + \lceil \log((n-1)!) \rceil \Big) \\
        <& (1 - p) \big(\log(n!) + 3 \big) + p \big(\log(n) - 9/\delta + \log((n-1)!) + 3 \big) \\
        = & \log(n!) + 3 - 9 p /\delta.
    \end{align*}
    where we used $m := \lfloor 2^{-9/\delta} n\rfloor$ in the second line. Since the encoding length must be at least the Shannon entropy $\log(n!)$ of the encoded permutation $\pi$, we conclude that $\Pr(F) = p \leq \delta/3$.

    Putting the two bounds together, we have $\Pr(E \cup F) \leq \Pr(E) + \Pr(F) \leq 2\delta/5$.
\end{proof}

Where does the game come into play? The reader might find it useful to recall Figure~\ref{fig:polygons}. The space $\Z^2$ is split into $\ell$ blocks horizontally, where each block $j \in [\ell]$ holds information about the $j$-th bit of the input vector. Let us focus on a particular block $j$. As $\Alg$ runs, some of its queries ``hit'' this block while others don't. Precisely, we say that a $\Sample(X_i)$ query \emph{hits $j$} if the random block chosen in Algorithm~\ref{alg:oracle} is $j$. Likewise, we say that a $\Contains((a,b), X_i)$ query \emph{hits $j$} if $\lfloor a/n \rfloor = j$. Now we restrict ourselves to the queries hitting $j$ and ignore all the others. Then, over time, $\Alg$ (as Arthur) is exactly playing the game with us (as Merlin) on the secret permutation $\pi_j$: $\Sample()$ queries correspond to $\Random()$ requests; $\Contains()$ queries correspond to $\Probe()$ requests; Algorithm~\ref{alg:oracle} branches to the ``otherwise'' clauses if and only if the answer is ``yes''. Whatever information $\Alg$ collects in the queries hitting other blocks $j'$, it has no effect on the decision in this game, as all other permutations are independent of $\pi_j$, and as $x,y$ do not play a role in this game.

With this in mind, let random variable $N_j$ count the number of queries hitting $j$. Define
\[
    J := \{j \in [\ell] : x_j \text{ is accessed and } N_j \leq m \} \quad \text{and} \quad J' := \{j \in [\ell] : N_j > m \}.
\]
where $m := \lfloor 2^{-9/\delta} n \rfloor$. By Lemma \ref{lem:mini-game}, $\Pr(j \in J) \leq 2\delta/5$ for every $j \in [\ell]$, so $\E[|J|] \leq 2\delta\ell/5$. On the other hand, by definition of $J'$ and our assumption that the number of queries is at most $\frac{\delta}{360 \cdot 2^{9/\delta}} \cdot \frac{n}{\eps^2}$, we can bound
\[
    |J'|
    < \frac{1}{m+1} \cdot \frac{\delta}{360 \cdot 2^{9/\delta}} \cdot \frac{n}{\eps^2}
    \leq \frac{\delta\ell}{10}.
\]
So Algorithm~\ref{alg:oracle} accesses less than $\frac{2\delta\ell}{5} + \frac{\delta\ell}{10} = \frac{\delta\ell}{2}$ bits of $x$ in expectation. Symmetrically, the same bound holds for $y$. Altogether it accesses less than $\delta\ell$ bits of $x,y$. But as we argued, the output can be used to distinguish $\langle x, y \rangle \geq \sqrt{\ell}$ and $\langle x, y \rangle \leq -\sqrt{\ell}$. This contradicts Lemma~\ref{lem:gap-hamming}.

\section{Reduction from discrete to continuous space}
\label{sec:discrete-continuous}
Union volume estimation and Klee's measure problem were formulated in the continuous space $\R^d$. Their analogs in discrete space $\Z^d$ are only simpler due to the following reduction. Consider the natural embedding $\varphi$ that blows up each point $x = (x_1, \dots, x_d) \in \Z^d$ into a continuous cube $\varphi(x) = [x_1, x_1+1] \times \cdots \times [x_d, x_d+1] \subset \R^d$. For any finite subset $X \subset \Z^d$, denote its image by $\varphi(X) \subset \R^d$. We have:
\begin{itemize}
    \item $\Vol(\varphi(X)) = |X|$.
    \item To sample from $\varphi(X)$ uniformly, we can first sample $x \in X$ uniformly and then sample from $\varphi(x)$ uniformly.
    \item If $X$ is a discrete box then $\varphi(X)$ is a continuous box, and vice versa.
\end{itemize}

Given this correspondence, our algorithm for Klee's measure problem can be made to handle discrete boxes in $\Z^d$. On the other hand, the discrete objects in our lower bound construction can be realized as axis-aligned polygons in $\R^2$.

\bibliography{main}

\end{document}